\documentclass[prl,reprint,showpacs]{revtex4-1}
\usepackage[T1]{fontenc}
\usepackage[latin9]{inputenc}
\usepackage{color}
\usepackage{amsthm}
\usepackage{amsmath}
\usepackage{amssymb}
\usepackage{graphicx}
\usepackage[normalem]{ulem}  % for \sout

\theoremstyle{plain}
\newtheorem{thm}{\protect\theoremname}
\providecommand{\theoremname}{Theorem}

\begin{document}
\global\long\global\long\global\long\def\bra#1{\mbox{\ensuremath{\langle#1|}}}
\global\long\global\long\global\long\def\ket#1{\mbox{\ensuremath{|#1\rangle}}}
\global\long\global\long\global\long\def\bk#1#2{\mbox{\ensuremath{\ensuremath{\langle#1|#2\rangle}}}}
\global\long\global\long\global\long\def\kb#1#2{\mbox{\ensuremath{\ensuremath{\ensuremath{|#1\rangle\!\langle#2|}}}}}

\theoremstyle{plain}
\newtheorem*{remark}{Remark}

\def\cblue{\color{blue}}
\def\cred{\color{red}}
\def\cmag{\color{magenta}}
\def\cblack{\color{black}}
\def\cmh{\cmag}

\title{Creation of superposition of unknown quantum states}

\author{Micha\l{} Oszmaniec}

\email{michal.oszmaniec@icfo.es}

\affiliation{Center for Theoretical Physics, Polish Academy of Sciences, Al. Lotników
32/46, 02-668 Warszawa}

\affiliation{ICFO - Institut de Ciencies Fotoniques, Mediterranean Technology Park,
08860 Castelldefels (Barcelona), Spain }

\author{Andrzej Grudka}

\affiliation{Faculty of Physics, Adam Mickiewicz University, 61-614 Poznan, Poland}

\author{Micha\l{} Horodecki}

%\affiliation{Institute of Theoretical Physics and Astrophysics, University of
%Gda\'nsk, 80-952 Gdañsk, Poland}
%
%
%\affiliation{National Quantum Information Centre of Gda\'nsk, 81-824 Sopot, Poland}

\affiliation{Faculty of Mathematics, Physics and Informatics, 
Institute of Theoretical Physics and Astrophysics, University of
Gda\'nsk, 80-952 Gdañsk, Poland}

\author{Antoni W\'{o}jcik}

\affiliation{Faculty of Physics, Adam Mickiewicz University, 61-614 Poznan, Poland}

\begin{abstract}

The superposition principle is one of the landmarks of quantum mechanics. The importance of quantum superpositions provokes questions about the limitations that quantum mechanics itself imposes on the possibility of their generation. In this work we systematically study the problem of creation of superpositions of unknown quantum states. First, we prove a no-go theorem that forbids the existence
of a universal probabilistic quantum protocol producing a superposition of two
unknown quantum states. Secondly, we provide an explicit probabilistic
protocol generating a superposition of two unknown states, each having
a fixed overlap with the known referential pure state.The protocol can be applied to generate coherent superposition
of results of independent runs of subroutines in a quantum computer. Moreover, in the context
of quantum optics it can be used to efficiently  generate highly nonclassical states or non-gaussian states.
\end{abstract}

\keywords{quantum information, quantum superpositions, foundations of quantum
mechanics}

\maketitle

%\section*{Introduction}

The existence of superpositions of pure quantum states is one of the
most intriguing consequences of the postulates of quantum mechanics.
Quantum superpositions are crucial for the path-integral formulation
of quantum mechanics \cite{Feynman2005} and are responsible for
numerous nonclassical phenomena that are considered to be the key
features of quantum theory \cite{Schroedinger1935}. The prominent
examples are: quantum interference \cite{Anderson1998,Zeilinger1999,Wineland2013}
and quantum entanglement \cite{Schroedinger1935a}. Coherent addition of wavefunctions is also responsible for quantum coherence, a feature of quantum states that recently received a lot of attention  \cite{Baumgratz2014,Streltsov2015,Winter2015}.  Quantum superpositions
 are not only important from the foundational point
of view but also a feature of quantum mechanics that underpins the
existence of ultra-fast quantum algorithms (such as Shor factoring
algorithm \cite{Shor1997} or Grover search algorithm \cite{Grover1996}),
quantum cryptography \cite{Gisin2002} and efficient quantum metrology
\cite{Giovannetti2011}. 

The importance of quantum superpositions provokes questions about
the restrictions that quantum mechanics itself imposes on the possibility
of their generation. Studies of the limitations of the possible operations
allowed by quantum mechanics have a long tradition are important both
from the fundamental perspective as well as for the applications in
quantum information theory. On one hand quantum mechanics offers a
number of protocols that either outperform all existing classical
counterparts or even allow to perform tasks that are impossible in
the classical theory (such as quantum teleportation \cite{Bennett1993}).
On the other hand a number of no-go theorems \cite{Wootters1982,Dieks1982,Barnum1996,Pati2000,Pati2002,Piani2008}
restrict a class of protocols that are possible to realise within
quantum mechanics. Finally, such no-go theorems can be themselves
useful for practical purposes. For instance a no-clonning theorem
can be used to certify the security of quantum cryptographic protocols
\cite{Gisin2002}. 

In this paper, we consider the scenario in which we are given two unknown pure quantum states and our task
is to create, using the most general operations allowed by quantum
mechanics, their superposition with some complex weights. 
Essentially the same question was posed in a parallel work of Alvarez-Rodriguez et al. \cite{Alvarez-Rodriguez2014}:
namely the authors asked about the existence of {\it quantum adder} - a machine, that would superpose two registers with the plus sign.

Here, we first prove a no-go theorem, showing that it is impossible to create superposition of two unknown states.
We discuss the relation of our theorem with the no-go results of \cite{Alvarez-Rodriguez2014}.
Subsequently, we provide a protocol that probabilistically creates superposition of two states
having fixed nonzero overlaps with some referential state.  We show that, by using appropriate encoding, the protocol can be used to generate superpositions of unknown vectors from the subspace perpendicular to the referential state, thus allowing for generation of coherent superpositions of the results of quantum subroutines of a given quantum algorithm. This actually shows how to circumvent 
our no-go theorem to some extent. 
We also discuss optical implementation of the protocol, with the referential state being the vacuum state. 
Finally, we discuss the differences between our results, and analogous results concerning cloning.

\paragraph{Introduction}
Before we proceed we need to carefully analyse the concept of quantum superpositions.
Recall first that the global phase of a wavefunction is not a physically
accessible quantity. This redundancy can be removed when one interprets
pure states as one dimensional orthogonal projectors acting on the
relevant Hilbert space. In what follows the pure state corresponding
to a normalized vector $\ket{\psi}$ will be denoted by %
\footnote{We decided to use this notation to avoid ambiguity. The relation to
the standard notation used in quantum mechanics is simply given by
$\mathbb{P}_{\psi}=\kb{\psi}{\psi}$.
} $\mathbb{P}_{\psi}$. Normalized vectors that
rise to the same pure state $\mathbb{P}_{\psi}$ are called vector
representatives of $\mathbb{P}_{\psi}$. They are defined up to a
global phase i.e. $\mathbb{P}_{\psi}=\mathbb{P}_{\psi'}$ if and only
if $\ket{\psi'}=\mathrm{exp}\left(i\theta\right)\ket{\psi}$, for
some phase $\theta$. Let now $\alpha,\beta$ be complex numbers satisfying $|\alpha|^2+|\beta|^2=1$ and let  $\mathbb{P}_{\psi}$, $\mathbb{P}_{\phi}$
be two pure states. By  $\mathbb{P}_{\alpha,\beta}\left(\ket{\psi},\ket{\phi}\right)$ we 
denote the projector onto the superposition of $\ket{\psi}$ and $\ket{\phi}$
\begin{equation}
\mbox{\ensuremath{\mathbb{P}_{\alpha,\beta}\left(\ket{\psi},\ket{\phi}\right)}}=\mathbb{P}_{\Psi},\,\ket{\Psi}=\mathcal{N}^{-1}\cdot\left(\alpha\ket{\psi}+\beta\ket{\phi}\right)\,,\label{eq:superposition def}
\end{equation}
where $\mathcal{N}=\sqrt{1+\mathrm{2}\cdot\mathrm{Re}\left(\bar{\alpha}\beta\bk{\psi}{\phi}\right)}$
is a normalization factor. The crucial observation is that $\mathbb{P}_{\alpha,\beta}\left(\ket{\psi},\ket{\phi}\right)$
is not a well-defined function of the states $\mathbb{P}_{\psi}$
and $\mathbb{P}_{\phi}$. This is because $\mathbb{P}_{\alpha,\beta}\left(\ket{\psi},\ket{\phi}\right)$
depends on vector representatives $\ket{\psi},\,\ket{\phi}$ , whose
phases can be gauged independently. Consequently, we have the infinite
family of pure states
\begin{equation}
\mathbb{P}_{\alpha,\beta}\left(\ket{\psi},\mathrm{exp}\left(i\theta\right)\ket{\phi}\right)\,,\,\theta\in\left[0,2\pi\right)\,,\label{eq:infinite family}
\end{equation}
which can be legitimately called superpositions of $\mathbb{P}_{\psi}$
and $\mathbb{P}_{\phi}$. This phenomenon appears already in the simplest
example of a qbit. For $\mathbb{P}_{\psi}=\kb 00,\ \mathbb{P}_{\phi}=\kb 11$
and $\alpha=\beta=\frac{1}{\sqrt{2}}$  the family given by
(\ref{eq:infinite family}) can be identified with the equator on the
Bloch ball. 
The analogous analysis was conducted in \cite{Alvarez-Rodriguez2014} and it was argued there that the ambiguity of the relative phase forbids the existence of the universal quantum adding machine. In our approach we propose to relax the definition of superposing,
so that it is not excluded from the very definition.  However, we will still prove a no-go theorem. 

We now settle the notation that we will used throughout the article.
By $\mathrm{Herm}\left(\mathcal{H}\right)$, and $\mathcal{D}\left(\mathcal{H}\right)$
we denote respectively sets of hermitian operators and the set of
density matrices on Hilbert space $\mathcal{H}$. By $\mathcal{CP}\left(\mathcal{H},\mathcal{K}\right)$
we denote the set of completely positive (CP) maps $\Lambda:\mathrm{Herm}\left(\mathcal{H}\right)\rightarrow\mathrm{Herm}\left(\mathcal{K}\right)$ ($\mathcal{K}$ is some arbitrary Hilbert space). 

Let us now formalise our scenario.  We
assume that we have access to two identical quantum registers (to
each of them we associate a Hilbert space $\mathcal{H}$) and we know
that the input state is a product of unknown pure states $\mathbb{P}_{\psi}\otimes\mathbb{P}_{\phi}$.
Our aim is to generate from this input the superposition $\mathbb{P}_{\alpha,\beta}\left(\ket{\psi},\ket{\phi}\right)$
 by the most general operations
allowed by quantum mechanics . By such operations we understand the
application of a quantum channel between $\mathcal{H}^{\otimes2}$
and $\mathcal{H}$ , followed by the postselection conditioned on
the result of some generalized measurement \cite{Nielsen2010}. This class of operations has a convenient
mathematical characterization. It consists of CP maps $\Lambda\in\mathcal{CP}\left(\mathcal{H}^{\otimes2},\mathcal{H}\right)$
that do not increase the trace i.e. $\mathrm{tr}\left[\Lambda\left(\rho\right)\right]\leq\mathrm{tr}\left(\rho\right)$
for all $\rho\in\mathcal{D}\left(\mathcal{H}^{\otimes2}\right)$.
For a given state $\rho$ the number $\mathrm{tr}\left[\Lambda\left(\rho\right)\right]$
is the probability that the operation $\Lambda$ took place. If the
operation takes place, the state $\rho$ undergoes the transformation
$\rho\rightarrow\frac{\Lambda\left(\rho\right)}{\mathrm{tr}\left[\Lambda\left(\rho\right)\right]}\,.$

%\begin{figure}[h]
%\begin{centering}
%\includegraphics[width=5cm]{process}
%\par\end{centering}
%\protect\caption{\label{fig:superp}\cred{\sout{Illustration of the considered setting.
%We require  that for all input states of the form $\mathbb{P}_{\psi}\otimes\mathbb{P}_{\phi}$ the quantum operation $\Lambda$ produces, when it
%is successful, the superposition ${\mathbb{P}_{\alpha,\beta}\left(\protect\ket{\psi},\protect\ket{\phi}\right)}$.}}}
%\end{figure}
\paragraph{No-go theorem}
We prove the no-go result in the strongest possible
form. First, we impose the minimal assumptions on the generated superpositions,
assuming only that vectors $\ket{\psi},\,\ket{\phi}$ are vector representatives
depending on the input states (in other words we are not interested
in the relative phase $\theta$ of the superposition appearing in
(\ref{eq:infinite family})). Secondly, we allow the probabilistic
protocols, i.e. the superposition may be created with some probability.
\begin{thm}
\label{nogo theorem}Let $\alpha,\beta$ be nonzero complex numbers
satisfying $\left|\alpha\right|^{2}+\left|\beta\right|^{2}=1$ and
let $\mathrm{dim}\mathcal{H}\geq2$ . There exist no nonzero completely
positive map $\Lambda\in\mathcal{CP}\left(\mathcal{H}^{\otimes2},\mathcal{H}\right)$
such that for all pure states $\mathbb{P}_{1},\mathbb{P}_{2}$
\begin{equation}
\Lambda\left(\mathbb{P}_{1}\otimes\mathbb{P}_{2}\right)\propto\kb{\Psi}{\Psi}\,,\label{eq:sup condition}
\end{equation}
where
\begin{equation}
\ket{\Psi}=\alpha\ket{\psi}+\beta\ket{\phi}\label{eq:superp vector}
\end{equation}
and $\kb{\psi}{\psi}=\mathbb{P}_{1}$, $\kb{\phi}{\phi}=\mathbb{P}_{2}$
and the representants $\ket{\psi},\ket{\phi}$ may in general depend
on both $\mathbb{P}_{1}$ and $\mathbb{P}_{2}$.\end{thm}
\begin{remark}
In particular, for two pairs $\left(\mathbb{P}_{1},\mathbb{P}_{2}\right)$
and $\left(\mathbb{P}_{1},\mathbb{P}'_{2}\right)$ the representant
of $\mathbb{P}_{1}$ can be different for each pair.
\end{remark}
\begin{proof}[Sketch of the proof]
Assume that there exist a nonzero CP map $\Lambda$ satisfying (\ref{eq:sup condition}).
Let the collection of operators $\left\{ V_{i}\right\} _{i\in I}\,,\, V_{i}:\mathcal{H}^{\otimes2}\rightarrow\mathcal{H}$
, form the Kraus decomposition \cite{Nielsen2010} of $\Lambda$,
$\Lambda\left(\rho\right)=\sum_{i\in I}V_{i}\rho V_{i}^{\dagger}$.
Since operators $\lambda\kb{\Psi}{\Psi}$, $\lambda\geq0$, belong
to the extreme ray of the cone of nonnegative operators on $\mathcal{H}$
we must have 
\begin{equation}
V_{i}\mathbb{P}_{1}\otimes\mathbb{P}_{2}V_{i}^{\dagger}\propto\kb{\Psi}{\Psi}\,,\,\text{for all}\, i\in I.\label{eq:single Krauss.}
\end{equation}
Consequently, it is enough to consider only CP maps that have one
operator in their Kraus decomposition. In such case (\ref{eq:sup condition})
reduces to the investigation of a single linear operator. If (\ref{eq:single Krauss.})
is satisfied then it necessary must hold for $\mathbb{P}_{1},\,\mathbb{P}_{2}$
having support on two dimensional subspaces of $\mathcal{H}$. Therefore,
it suffices to show that in the qbit case only operators $V_{i}$
that satisfy condition (\ref{eq:single Krauss.}) are the null operators.
We present the proof of this in the Supplemental Material \cite{supp}.
The main difficulty of the proof stems from the fact that the condition
(\ref{eq:single Krauss.}) is non-linear in the input state $\mathbb{P}_{1}\otimes\mathbb{P}_{2}$. 
\end{proof}
Theorem \ref{nogo theorem} shows that, even if we allow for postselection,
there exist no quantum operation that produces superpositions of all unknown
pure quantum states with some probability (we allowed this probability
to be zero for some pairs of input states and in general it can be
different for different inputs). 
 We would like to stress
that the creation of superpositions is still impossible even if we allow
for the arbitrary dependence of the relative phase of the input states. Namely,
in our formulation of the problem we explicitly assumed 
that vector representatives $\ket{\psi},\,\ket{\phi}$ of states $\mathbb{P}_{\psi}$
and $\mathbb{P}_{\phi}$ are some functions of these states. As a matter of fact, otherwise 
one would not be able to formulate the problem of generation of superpositions
in a consistent manner. 
 We  emphasize that in that respect the problem of creation of superpositions
is different to quantum cloning 
\footnote{In quantum cloning one seeks of a physical transformation that would
realize the (well-defined) mapping $\kb{\psi}{\psi}\rightarrow \kb{\psi}{\psi}^{\otimes k}$,
which is a well-defined function.}.  Moreover, to our best knowledge, there is no
immediate connection between the no-cloning theorem \cite{Wootters1982,Dieks1982}
and its generalized variants (such as no-deleting theorem \cite{Pati2000}
or no-anticloning theorem \cite{Pati2002}) to our result. This is
a consequence of the fact that $\Lambda$ must be non-invertible and
therefore cannot be used to obtain a cloning map. Moreover, in the
formulation of the theorem we allow for situations in which for some input
states $\mathbb{P}_{\psi}\otimes\mathbb{P}_{\phi}$ the probability
of success is zero. 
\paragraph{Constructive protocol}
It is  natural to study whether it is possible to
create quantum superpositions if we have some knowledge about the
input states. Except for specifying the class of input states for
which a given protocol would work, it is also necessary to prescribe
precisely which superpositions will be generated (see discussion before
Eq.(\ref{eq:infinite family})). In what follows we present an explicit
protocol that generates superpositions of unknown pure states $\mathbb{P}_{\psi},\,\mathbb{P}_{\phi}$
having fixed nonzero overlaps with some referential pure state $\mathbb{P}_{\chi}$
(see Figure \ref{fig:bloch}). Let us describe the superpositions
that will be generated by our protocol. Let $\ket{\chi}$ be a vector
representative of $\mathbb{P}_{\chi}$. For every pair of normalised
vectors $\ket{\psi},\ket{\phi}$ satisfying $\bk{\chi}{\psi}\neq0,\,\bk{\chi}{\psi}\neq0$
we define their superposition
\begin{equation}
\ket{\Psi}=\alpha\frac{\bk{\chi}{\phi}}{\left|\bk{\chi}{\phi}\right|}\ket{\psi}+\beta\frac{\bk{\chi}{\psi}}{\left|\bk{\chi}{\psi}\right|}\ket{\phi}\,,\label{eq:proper superpositions}
\end{equation}
The norm of this vector is given by
\begin{equation}
\mathcal{N}_{\Psi}=\sqrt{1+2\cdot\mathrm{Re}\left(\bar{\alpha}\beta\frac{\mathrm{tr}\left(\mathbb{P}_{\ket{\chi}}\mathbb{P}_{\ket{\psi}}\mathbb{P}_{\ket{\phi}}\right)}{\left|\bk{\chi}{\phi}\right|\left|\bk{\chi}{\psi}\right|}\right)}\,.\label{eq:norm constant2}
\end{equation}
The vector $\ket{\Psi}$ changes only by a global phase once any of
the vectors $\ket{\psi},\ket{\phi},\,\ket{\chi}$ gets multiplied
by a phase factor. Consequently, $\mathbb{P}_{\Psi}$ can be regarded
as well-defined function of the states $\mathbb{P}_{\ket{\psi}},\,\mathbb{P}_{\ket{\phi}}$,
provided they have nonzero overlap with $\mathbb{P}_{\chi}$. This
can be also seen from the explicit formula,
\begin{align}
\kb{\Psi}{\Psi} & =\left|\alpha\right|^{2}\mathbb{P}_{\psi}+\left|\beta\right|^{2}\mathbb{P}_{\phi}+\nonumber \\
 & +\left(\alpha\beta^{\ast}\frac{\mathbb{P}_{\psi}\mathbb{P}_{\chi}\mathbb{P}_{\phi}}{\sqrt{\mathrm{tr}\left(\mathbb{P}_{\psi}\mathbb{P}_{\chi}\right)}\sqrt{\mathrm{tr}\left(\mathbb{P}_{\phi}\mathbb{P}_{\chi}\right)}}+h.c\right)\,.\label{eq:full formula projector}
\end{align}
One could argue that the above choice of the superposition $\kb{\Psi}{\Psi}$ 
is somewhat arbitrary. However, the mapping $\left(\mathbb{P}_{\psi},\mathbb{P}_{\phi}\right)\rightarrow\kb{\Psi}{\Psi}$ is related to the so-called Pancharatnam
connection and appears in studies concerning the superposition rules
from the perspective of geometric approach to quantum mechanics \cite{Ercolessi2010,Manko2002}.
Moreover, it shown in \cite{Chaturvedi2013} that Eq.\eqref{eq:proper superpositions}
has a strong connection with the concept of the geometric phase. Finally,
from the purely operational grounds, Eq.\eqref{eq:proper superpositions}
constitute a rightful superposition of states $\mathbb{P}_{\psi},\,\mathbb{P}_{\phi}$
and as we vary coefficients $\alpha,\beta$ we can recover all possible
superpositions of $\mathbb{P}_{\psi},\,\mathbb{P}_{\phi}$. 
\begin{figure}[h]
\begin{centering}
\includegraphics[width=2.4cm]{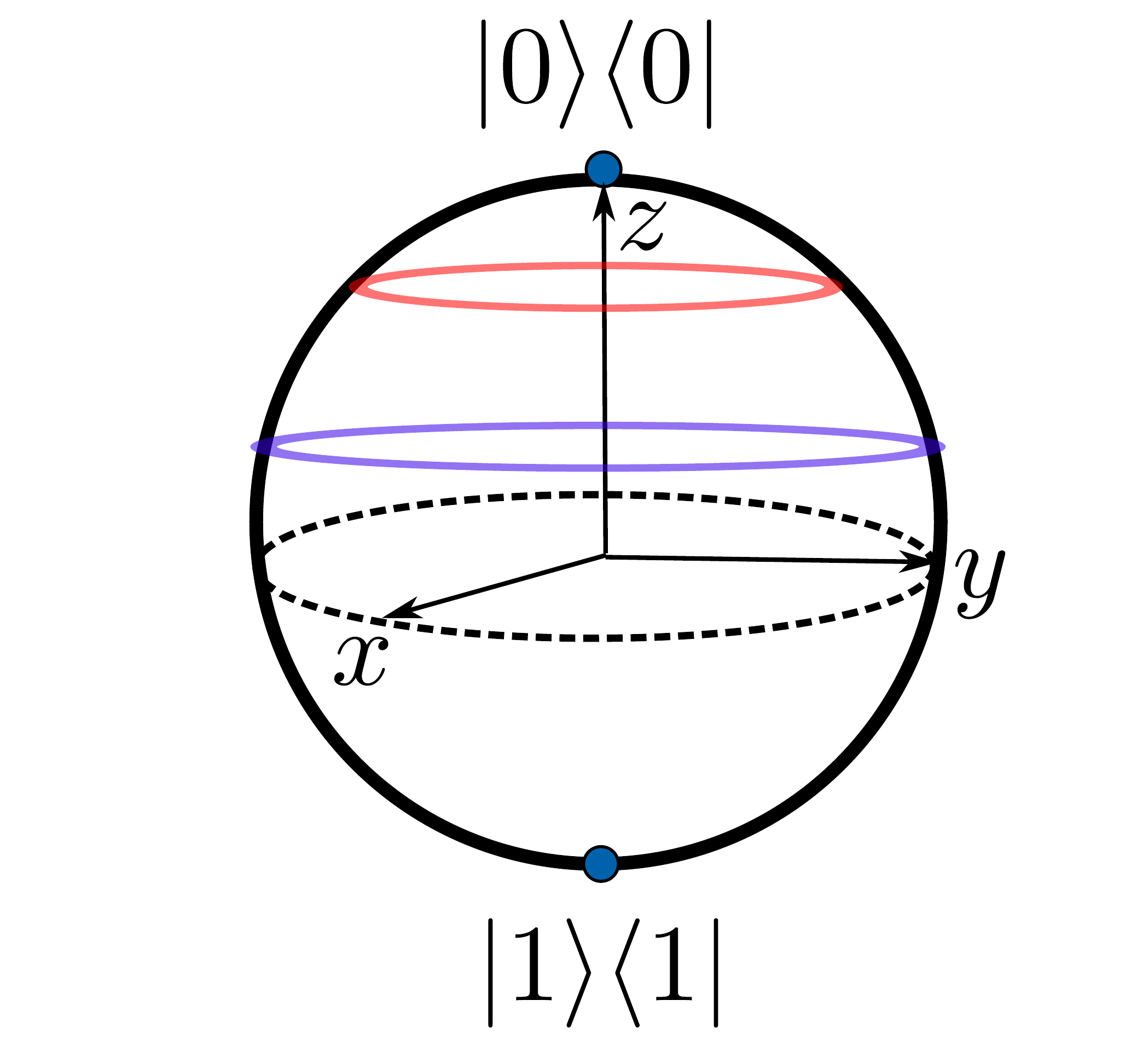}
\par\end{centering}
\protect\caption{\label{fig:bloch}Graphical representation of the class of input states
satisfying $\mathrm{tr}\left(\mathbb{P}_{\chi}\mathbb{P}_{\psi}\right)=c_{1}\,,\,\mathrm{tr}\left(\mathbb{P}_{\chi}\mathbb{P}_{\phi}\right)=c_{2}$
for $\mathcal{H}=\mathbb{C}^{2}$. For convenience we set $\mathbb{P}_{\chi}=\protect\kb 00$.}
\end{figure}
\begin{thm}
\label{protoco qbit} Let $\mathbb{P}_{\chi}$ be a fixed pure state
on Hilbert space $\mathcal{H}$. There exist a CP map $\Lambda_{sup}\in\mathcal{CP}\left(\mathbb{C}^{2}\otimes\mathcal{H}^{\otimes2},\mathcal{H}\right)$
such that for all pure states $\mathbb{P}_{\psi},\,\mathbb{P}_{\phi}$
on $\mathcal{H}$ satisfying 
\begin{equation}
\mathrm{tr}\left(\mathbb{P}_{\chi}\mathbb{P}_{\psi}\right)=c_{1}\,,\,\mathrm{tr}\left(\mathbb{P}_{\chi}\mathbb{P}_{\phi}\right)=c_{2}\,,\label{eq:conditions}
\end{equation}
 we have
\begin{equation}
\Lambda_{sup}\left(\mathbb{P}_{\nu}\otimes\mathbb{P}_{\psi}\otimes\mathbb{P}_{\phi}\right)\propto\kb{\Psi}{\Psi}\,,\label{eq:protocole result}
\end{equation}
where $\mathbb{P}_{\nu}$ , $\ket{\nu}=\alpha\ket 0+\beta\ket 1$,
is an unknown qbit state and the vector $\ket{\Psi}$ is given by
(\ref{eq:proper superpositions}). Moreover, a CP map $\Lambda_{sup}$
realising (\ref{eq:protocole result}) is unique up scaling.\end{thm}
\begin{proof}
We first present a protocol that realizes (\ref{eq:protocole result}).\textcolor{black}{{}
Let us define and auxiliary normalized qbit vector $\ket{\mu}=\mathcal{C}\cdot\left(\sqrt{c_{1}}\ket 0+\sqrt{c_{2}}\ket 1\right)$,
where $\mathcal{C}$ is a normalization constant. We set }$\Lambda_{sup}=\Lambda_{4}\circ\Lambda_{3}\circ\Lambda_{2}\circ\Lambda_{1}$,
where
\begin{eqnarray}
\Lambda_{1}\left(\rho\right) & = & V_{1}\rho V_{1}^{\dagger}\,,\, V_{1}=\kb 00\otimes\mathbb{I}\otimes\mathbb{I}+\mathbb{\kb 11}\otimes\mathbb{S}\,,\\
\Lambda_{2}\left(\rho\right) & = & V_{2}\rho V_{2}^{\dagger}\,,\, V_{2}=\mathbb{I}\otimes\mathbb{I}\otimes\kb{\chi}{\chi}\,,\\
\Lambda_{3}\left(\rho\right) & = & V_{3}\rho V_{3}\,,\, V_{3}=\mathbb{P}_{\mu}\otimes\mathbb{I}\otimes\mathbb{I}\,,\\
\Lambda_{4}\left(\rho\right) & = & \mathrm{tr}_{13}\left(\rho\right)\,.
\end{eqnarray}
In the above $\mathbb{S}$ denotes the unitary operator that swaps
between two copies of $\mathcal{H}$ and $\mathrm{tr}_{13}\left(\cdot\right)$
is the partial trace over the first and the third factor in the tensor
product $\mathbb{C}^{2}\otimes\mathcal{H}\otimes\mathcal{H}$. For
a graphical presentation of the above protocol see Fig.\ref{fig:circut}.
Operation $\Lambda_{sup}$ is completely positive and trace
non-increasing. Direct calculation shows that under the assumed conditions
(\ref{eq:protocole result}) indeed holds. We prove the uniqueness
result in the Supplemental Material \cite{supp}.
\end{proof}
\begin{figure}[h]
\begin{centering}
\includegraphics[width=5cm]{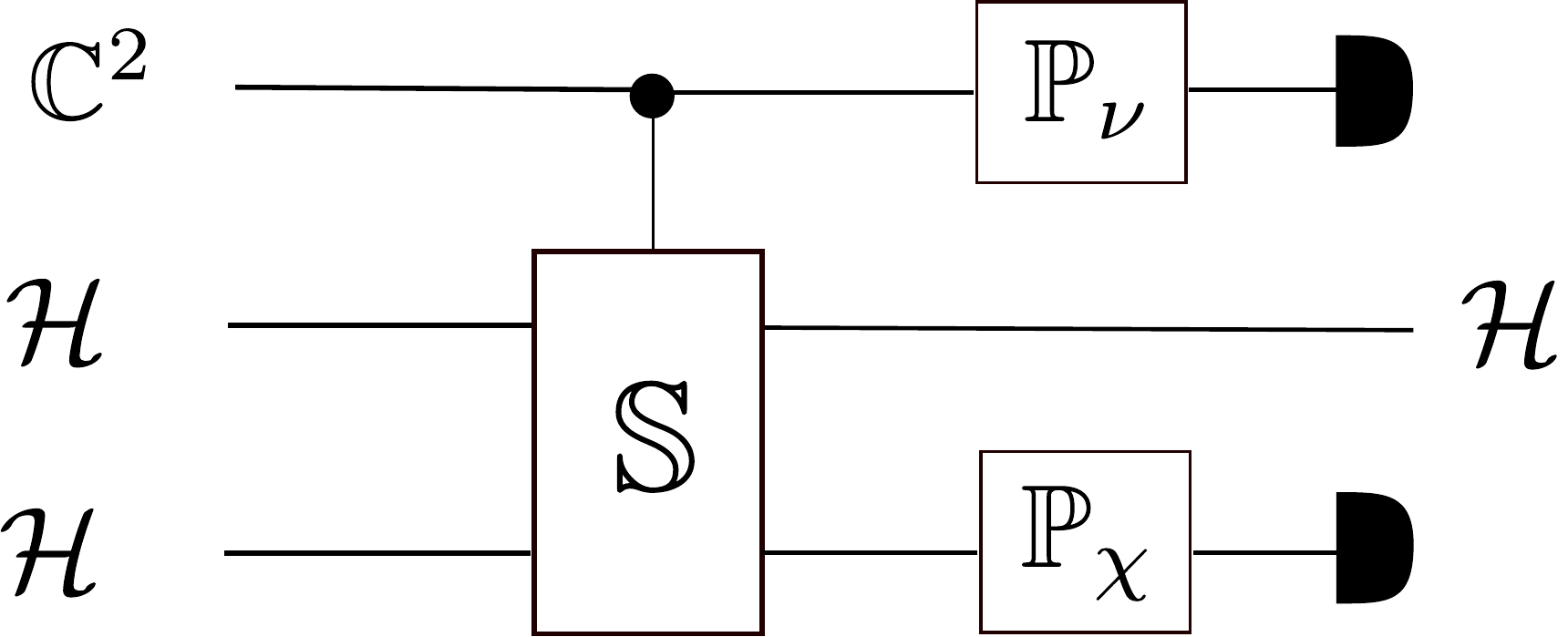}
\par\end{centering}
\protect\caption{\label{fig:circut}Graphical representation of the circut realising
the map $\Lambda_{sup}$.}
\end{figure}
The probability that the above protocol will successfully create superpositions
of states is given by
\begin{equation}
P_{succ}=\mathrm{tr}\left[\Lambda_{sup}\left(\mathbb{P}_{\nu}\otimes\mathbb{P}_{\psi}\otimes\mathbb{P}_{\phi}\right)\right]=\frac{c_{1}c_{2}}{c_{1}+c_{2}}\mathcal{N}_{\Psi}^{2}\,.\label{eq:p succ}
\end{equation}
The map $\Lambda_{sup}$ cannot be rescaled to increase the probability
of success. This follows from the (tight) operator inequality $\left(V_{3}V_{2}V_{1}\right)^{\dagger}\left(V_{3}V_{2}V_{1}\right)\leq\mathbb{I\otimes\mathbb{I\otimes}\mathbb{I}}$. Taking into account the uniqueness (up to scaling) of $\Lambda_{sup}$ we
get that $P_{succ}$ from (\ref{eq:p succ}) is the maximal achievable
probability of success (for inputs specified in the assumptions of
Theorem (\ref{protoco qbit})). However, for fixed coefficients $\alpha,\beta$
it is possible to design a CP map that can achieve higher probability
of success \cite{supp}. Moreover, it is possible to generalize  the protocol $\Lambda_{sup}$ to the situation when we have
of $d$ input states (having nonzero overlap with $\mathbb{P}_{\chi}$)
and coefficients of superposition are encoded in an unknown state
of a qdit \cite{supp}. 

The existence of the map $\Lambda_{sup}$ shows that the problem of
creating superpositions of quantum states differs greatly from the
cloning problem. Probabilistic quantum cloning of pure states is possible
if and only if we have a promise that the input states belong to the
family of states whose vector representatives form a linearly independent
set \cite{Duan1998}. Consequently, the aforementioned family of states
must be discrete. Our protocol shows that it is possible to probabilisticly
create superpositions from unknown quantum states belonging to uncountable
families of quantum states. 

\paragraph{Applications}

There exist deterministic circuits realizing
classical arithmetic operations (like addition, multiplication, exponentiation
etc.) on a quantum computer \cite{Vedral1996}. However, to our best
knowledge there exist no protocols realizing addition 
on vectors belonging to the Hilbert space responsible for the computation. We now present a method to generate coherent superposition of results of quantum computations.  Assume that $\alpha=\beta=\sqrt{c_1}=\sqrt{c_2}=\frac{1}{\sqrt{2}}$.   By setting the overlap of vector representatives of $\mathbb{P}_\psi$ and $\mathbb{P}_\phi$ with $\ket{\chi}$ to be positive we  get
\begin{equation} \label{eq:encoding}
\ket{\psi}=\frac{1}{\sqrt{2}} \ket{\chi}+\frac{1}{\sqrt{2}}\ket{\psi^\perp} \ , \ket{\phi}=\frac{1}{\sqrt{2}} \ket{\chi}+\frac{1}{\sqrt{2}}\ket{\phi^\perp} \ ,
\end{equation}
where unit vectors $\ket{\psi^\perp},\ket{\phi^\perp}$ are perpendicular to $\ket{\chi}$. Input  states $\mathbb{P}_\psi$, $\mathbb{P}_\phi$ are in one-to-one correspondence with the vectors $\ket{\psi^\perp},\ket{\phi^\perp}$. By the application of $\Lambda_{sup}$ it is possible to obtain a state having the (non-normalized) vector representative 
\begin{equation}\label{eq:encodedsum}
\ket{\Psi}=\ket{\chi}+\frac{1}{2}\left( \ket{\psi^\perp} +\ket{\phi^\perp}  \right) \ ,
\end{equation} 
with probability $P_{succ}=\frac{1}{4}\left(1+\frac{1}{4}\left\|\ket{\psi^\perp}+\ket{\phi^\perp}\right\|^2\right)\geq \frac14$.
We have obtained a state encoding the superposition of \textit{unknown} vectors $\ket{\psi^\perp}, \ket{\phi^\perp}$ encoded in states $\mathbb{P}_\psi$ and $\mathbb{P}_\phi$ respectively.  The method presented above effectively superposes the wavefunctions coherently, provided one has access to the auxiliary one dimensional subspace (spanned by $\ket{\chi}$).  It is highly unexpected but by changing the perspective and by treating as "primary" objects the vectors perpendicular to $\ket{\chi}$ we have managed to effectively get around the  no-go result from Theorem 1. To apply the above protocol, one has to run quantum computation in a the perpendicular space. 
In Supplemental Material \cite{supp}, we present an exemplary scheme implementing such computation . 

The protocol $\Lambda_{sup}$ can be also used 
 to generate nonclassical sates in the context of quantum
optics. Let the states $\mathbb{P}_{\psi},\mathbb{P}_{\phi}$ describe
quantum fields in two different optical modes. Hilbert spaces associated
each of the modes are isomorphic and can be identified with the single-mode
bosonic Fock space. Moreover, let the auxiliary qbit be encoded in
a polarization of a single photon in different optical mode or in another
two level physical system. In such a setting the natural choice of the state
$\mathbb{P}_{\chi}$ is the Fock vacuum $\kb{0_{F}}{0_{F}}$ describing
the state of the field with no photons. As an input we can put coherent
or pure Gaussian states \cite{Mandel1995} that have fixed overlaps
with the vacuum. Then, the protocol $\Lambda_{sup}$ generically creates
highly nonclassical or respectively non-gaussian states. Operations
$\Lambda_{2},\Lambda_{3},\Lambda_{4}$ are relatively easy to realize
in this setting. The most demanding operation is the conditional swap
$\Lambda_{1}$. However, conditional swap can be realized in the optical
setting via implementation of phase flip operation and standard beam
splitters \cite{Filip2002}. The phase flip operation on the other
hand can in principle \cite{Jeong2014} be obtained in the optical
setting by coupling light to atoms inside the cavity, trapped ions,
or by the usage of cross-Kerr nonlinearities in materials with electromagnetically
induced transparency. Despite the possible difficulties with the implementation
the map $\Lambda_{sup}$ is  worth realizing  as it gives the maximal
probability of success. Moreover, the protocol $\Lambda_{sup}$ is
universal and can be used in different physical scenarios.

\paragraph{Discussion}
There is a number of open questions we did not adress here. First
of all, the relation of our no-go theorem to other no-go results in
quantum mechanics is not clear and requires further investigation.
The constructive protocol presented by us suggest a connection with
the recent works concerning the problem of controlling an unknown
unitary operation \cite{Zhou2011,Thompson2013,Feix2014,Bisio2015} (the referential pure
states $\mathbb{P}_{\chi}$ can be regarded as an analogue of the
known eigenvector of the ``unknown'' operation $U$ which allows
for its control). It would be also natural to study the problem of
approximate generation of quantum superpositions in (in analogy to
the problem of approximate cloning \cite{Werner1998}). Another
possible line of research is to investigate the probabilistic
protocols designed especially to generate superpositions of states
naturally appearing in the experimental context (like pure coherent
or Gaussian states).
\begin{acknowledgments}
We age grateful to Andreas Winter for suggesting the encoding scheme \eqref{eq:encoding}. We would also like to thank Remigiusz Augusiak, Konrad Banaszek, Rafa\l{}
Demkowicz-Dobrza\'{n}ski, Leonardo Guerini de Souza, Stefan Fischer, Ryszard Horodecki and Marek
Ku\'{s} for fruitful and inspiring discussions. We acknowledge the support of EU project SIQS, ERC grants CoG QITBOX and AdG QOLAPS, FOQUS the Generalitat de Catalunya (SGR 875) and Foundation for Polish Science TEAM project co-financed by
the EU European Regional Development Fund. M. O acknowledges the support of START scholarship granted by Foundation for Polish Science.
Part of this work was done in National Quantum Information Centre of Gda\'{n}sk.
\end{acknowledgments}

\bibliography{suplit}

\newpage

\section*{Supplemental material}
\appendix

\subsection*{Part A: No-go theorem}
\begin{proof}[Final step of the proof of Theorem 1]
We will show that there exist no nonzero linear mapping $V:\mathcal{\mathbb{C}}^{2}\otimes\mathbb{C}^{2}\rightarrow\mathbb{C}^{2}$
such that for all pairs of input pure qbit states $\mathbb{P}_{\psi},\mathbb{P}_{\phi}$
we have
\begin{equation}
V\mathbb{P}_{\psi}\otimes\mathbb{P}_{\phi}V^{\dagger}\propto\kb{\Psi}{\Psi}\,,\, \tag{S.1} \label{eq:qbit condition map}
\end{equation}
where

\begin{equation}
\ket{\Psi}=\alpha\ket{\psi}+\beta\ket{\phi}\, \tag{S.2},\label{eq:superp vector-1}
\end{equation}
where $\alpha,\beta$ are fixed nonzero complex numbers satisfying $\left|\alpha\right|^{2}+\left|\beta\right|^{2}=1$, and vector
representatives $\ket{\psi},\ket{\phi}$ are given by 

\begin{equation}
\ket{\psi}=\mathcal{F}\left(\mathbb{P}_{\psi},\mathbb{P}_{\phi}\right)\,,\,\ket{\phi}=\mathcal{G}\left(\mathbb{P}_{\psi},\mathbb{P}_{\phi}\right)\,, \tag{S.3}
\end{equation}
for some fixed functions $\mathcal{F},\mathcal{G}$. Let us fix the
standard product basis of $\mathbb{C}^{2}\otimes\mathbb{C}^{2}$ and
let us order it in a lexicographic order.
\[
\ket{v_{1}}=\ket 0\ket 0,\,\ket{v_{2}}=\ket 0\ket 1,\,\ket{v_{3}}=\ket 1\ket 0,\,\ket{v_{4}}=\ket 1\ket 1. \tag{S.4}
\]
Likewise, let us introduce the standard basis in the output Hilbert
space $\mathbb{C}^{2}$, $\ket{f_{1}}=\ket 0,\,\ket{f_{2}}=\ket 1$.
For such a choice of the basis the operator $V$ can be described as $2\times4$ matrix

\begin{equation}
V=\begin{pmatrix}a & b & c & d\\
e & f & g & h
\end{pmatrix}\,. \tag{S.5} \label{eq:operator V form}
\end{equation}
The condition (\ref{eq:qbit condition map}) can be written in the
form 

\begin{equation}
V\ket{\psi}\ket{\phi}=\mathcal{C}\left(\ket{\psi},\ket{\phi}\right)\left(\alpha\ket{\psi}+\beta\cdot\mathcal{Z}\left(\ket{\psi},\ket{\phi}\right)\ket{\phi}\right)\,,\tag{S.6} \label{eq:vectors condition}
\end{equation}
where $\ket{\psi},\ket{\phi}\in\mathbb{C}^{2}$ are unit vectors,
$\mathcal{C}\left(\ket{\psi},\ket{\phi}\right)$ is a complex-valued
function and $\mathcal{Z}\left(\ket{\psi},\ket{\phi}\right)$ is a
function taking values in the unit circle and satisfying $\mathcal{Z}\left(\ket{\psi},\ket{\phi}\right)=\mathcal{Z}\left(\mathrm{exp}\left(i\theta\right)\ket{\psi},\mathrm{exp}\left(i\theta\right)\ket{\phi}\right)$
, for all $\theta\in\mathbb{R}$. Taking $\ket{\psi}=\ket{\phi}$
in (\ref{eq:vectors condition}) we obtain 
\begin{equation}
V\ket{\psi}\ket{\psi}\perp\ket{\psi^{\perp}}\, \tag{S.7} ,\label{eq:perp condition}
\end{equation}
where $\ket{\psi^{\perp}}\in\mathbb{C}^{2}$ is an arbitrary vector
perpendicular to $\ket{\psi}$. Using (\ref{eq:perp condition}) for
\begin{equation}
\ket{\psi}\propto\ket 0+\alpha\ket 1\,,\,\ket{\psi^{\perp}}\propto-\bar{\alpha}\ket 0+\ket 1\,,\,\alpha\in\mathbb{C}\label{eq:indexation of vectors} \tag{S.8}
\end{equation}
we obtain 
\begin{equation}
V=\begin{pmatrix}f+g & b & c & 0 \\
0 & f & g & b+c
\end{pmatrix}\label{eq:simplified V} \ \tag{S.9} .
\end{equation}
Using the above we obtain that for every pair of vectors $\ket{\psi},\ket{\phi}\in\mathbb{C}^{2}$,

\begin{equation}
V\ket{\psi}\ket{\phi}=\bk{\chi_{1}}{\phi}\ket{\psi}+\bk{\chi_{2}}{\psi}\ket{\phi}\ ,\label{eq:simplification} \tag{S.10}
\end{equation}
where
\begin{gather}
\ket{\chi_{1}}=\bar{g}\ket 0+\bar{b}\ket 1\,, \tag{S.11}\\ 
\ket{\chi_{2}}=\bar{f}\ket 0+\bar{c}\ket 1\,. \tag{S.12}
\end{gather}
From (\ref{eq:simplification}) and (\ref{eq:vectors condition})
we obtain that for all unit vectors $\ket{\psi},\ket{\phi}\in\mathbb{C}^{2}$

{\small{}
\begin{equation}
\left(\alpha\cdot\mathcal{C}\left(\ket{\psi},\ket{\phi}\right)-\bk{\chi_{1}}{\phi}\right)\ket{\psi}+\left(\beta\cdot\tilde{\mathcal{C}}\left(\ket{\psi},\ket{\phi}\right)-\bk{\chi_{2}}{\psi}\right)\ket{\phi}=0\,,\label{eq:another condition} \tag{S.13}
\end{equation}
}where $\tilde{\mathcal{C}}\left(\ket{\psi},\ket{\phi}\right)=\mathcal{C}\left(\ket{\psi},\ket{\phi}\right)\mathcal{Z}\left(\ket{\psi},\ket{\phi}\right)$.
Setting in the above $\ket{\psi}=\ket{\chi_{2}^{\perp}}$ ($\ket{\chi_{2}^{\perp}}$ is
some unit vector perpendicular to $\ket{\chi_{2}}$) we obtain 
\begin{equation}
\left(\alpha\cdot\mathcal{C}\left(\ket{\chi_{2}^{\perp}},\ket{\phi}\right)-\bk{\chi_{1}}{\phi}\right)\ket{\chi_{2}^{\perp}}+\beta\cdot\tilde{\mathcal{C}}\left(\ket{\chi_{2}^{\perp}},\ket{\phi}\right)\ket{\phi}=0\,.\label{eq:final simplification} \tag{S.14}
\end{equation}
Since $\alpha,\beta\ne0$ and $\ket{\phi}$can be chosen in arbitrary
manner we obtain
\begin{gather}
\alpha\cdot\mathcal{C}\left(\ket{\chi_{2}^{\perp}},\ket{\phi}\right)-\bk{\chi_{1}}{\phi}=0\,, \tag{S.14}\\
\tilde{\mathcal{C}}\left(\ket{\chi_{2}^{\perp}},\ket{\phi}\right)=0\,, \tag{S.15}
\end{gather}
whenever $\ket{\chi_{2}^{\perp}}$ not linearly dependent with $\ket{\phi}$.
Consequently we obtain that $\mathcal{C}\left(\ket{\chi_{2}^{\perp}},\ket{\phi}\right)=\tilde{\mathcal{C}}\left(\ket{\chi_{2}^{\perp}},\ket{\phi}\right)=0$
and consequently $\bk{\chi_{1}}{\phi}=0$ for all $\ket{\phi}$ not
parallel to $\ket{\chi_{2}^{\perp}}$. Consequently we obtain $\ket{\chi_{1}}=0$.
An analogous argument shows that $\ket{\chi_{2}}=0$.
\end{proof}

\subsection*{Part B: Constructive protocols for superposition of two states}

In this part we complete the proof of Theorem 2
and derive the formulas for the maximal probability of success for
the generation of superposition $\mathbb{P}_{\Psi}$ via the usage
of CP maps $\Lambda\in\mathcal{CP}\left(\mathbb{C}^{2}\otimes\mathcal{H}^{\otimes2},\mathcal{H}\right)$. 
\begin{proof}[Proof of the uniqueness result from Theorem 2]
 Let $\ket{\nu}=\alpha\ket 0+\beta\ket 1$ and let $\mathbb{P}_{\psi},\mathbb{\mathbb{P}_{\chi}}$
be states on $\mathcal{H}$ satisfying $\mathrm{tr}\left(\mathbb{P}_{\psi}\mathbb{P}_{\chi}\right)=c_{1},$~$\mathrm{tr}\left(\mathbb{P}_{\phi}\mathbb{P}_{\chi}\right)=c_{2}$.
Let now $\Lambda\in\mathcal{CP}\left(\mathbb{C}^{2}\otimes\mathcal{H}^{\otimes2},\mathcal{H}\right)$
be the CP map satisfying 

\begin{equation}
\Lambda\left(\mathbb{P}_{\nu}\otimes\mathbb{P}_{\psi}\otimes\mathbb{P}_{\phi}\right)\propto\kb{\Psi}{\Psi}\,, \tag{S.16}
\end{equation}
where 
\begin{equation}
\ket{\Psi}=\alpha\frac{\bk{\chi}{\phi}}{\left|\bk{\chi}{\phi}\right|}\ket{\psi}+\beta\frac{\bk{\chi}{\psi}}{\left|\bk{\chi}{\psi}\right|}\ket{\phi}\, \tag{S.17} \label{eq:desired super}
\end{equation}
is the superposition of states we want to generate. Let $\left\{ V_{i}\right\} _{i\in I}\,,\, V_{i}:\mathbb{C}^{2}\otimes\mathcal{H}^{\otimes2}\rightarrow\mathcal{H}$
, form the Kraus decomposition \cite{Nielsen2010} of $\Lambda$.
Using the analogous argumentation to the one presented in the proof
of Theorem 1 we get that 
\begin{equation}
V_{i}\mathbb{P}_{\nu}\otimes\mathbb{P}_{\psi}\otimes\mathbb{P}_{\phi}V_{i}^{\dagger}\propto\kb{\Psi}{\Psi}\,,\,\text{for all}\, i\in I. \tag{S.18} \label{eq:second kraus}
\end{equation}
Let us focus on a single Kraus operator $V_{i}$. In what follows
we will drop the in index $i$ for simplicity. From (\ref{eq:second kraus})
we get 

\begin{equation}
V\ket{\nu}\ket{\psi}\ket{\phi}=\mathcal{A}\left(\alpha,\beta,\ket{\psi},\ket{\phi}\right)\cdot\left(\alpha\frac{\bk{\chi}{\phi}}{\left|\bk{\chi}{\phi}\right|}\ket{\psi}+\beta\frac{\bk{\chi}{\psi}}{\left|\bk{\chi}{\psi}\right|}\ket{\phi}\right)\, \tag{S.19},\label{eq:cration single craus}
\end{equation}
for all vectors $\ket{\psi},\,\ket{\phi}\in\mathcal{H}$ satisfying 
\begin{equation}
\left|\bk{\chi}{\psi}\right|=\sqrt{c_{1}}\,,\,\left|\bk{\chi}{\phi}\right|=\sqrt{c_{2}},\, \tag{S.20},
\end{equation} 
arbitrary $\ket{\nu}=\alpha\ket 0+\beta\ket 1$ and for $\mathcal{A}\left(\alpha,\beta,\ket{\psi},\ket{\phi}\right)$
being some unknown function. We will now show that condition (\ref{eq:cration single craus})
defines $V$ uniquely up to a multiplicative constant. Having this
result we will be able infer the uniqueness of $\Lambda$ (up to scaling).
Using the linearity of the left hand side of (\ref{eq:cration single craus})
in $\ket{\nu}$ we get 
\begin{equation}
\mathcal{A}\left(\alpha,\beta,\ket{\psi},\ket{\phi}\right)=\mathcal{A}\left(\ket{\psi},\ket{\phi}\right)\, \tag{S.21}.\label{eq:simplification2} 
\end{equation}
Moreover, from the linearity of $V$ and the condition (\ref{eq:cration single craus})
it follows that 
\begin{equation}
\mathcal{A}\left(\ket{\psi},\ket{\phi}\right)=\mathcal{A}\left(\mathrm{exp}\left(i\theta_{1}\right)\ket{\psi},\mathrm{exp}\left(i\theta_{2}\right)\ket{\phi}\right)\, \tag{S.22},\label{eq:phase in A}
\end{equation}
where $\theta_{1},\theta_{2}$ are arbitrary phases. Because of this
property it suffices to check the condition (\ref{eq:cration single craus})
for vectors of the form 

\begin{gather}
\ket{\psi}=\sqrt{c_{1}}\ket{\chi}+q_{1}\ket{\psi^{\perp}}\, \tag{S.23},\label{eq:vec repr 1}\\
\ket{\phi}=\sqrt{c_{2}}\ket{\chi}+q_{2}\ket{\phi^{\perp}}\,, \tag{S.24} \label{eq: vect repr 2}
\end{gather}
where $\ket{\chi}$ is some fixed vector representative of $\mathbb{P}_{\chi}$,
$q_{i}=\sqrt{1-c_{i}}$, and vectors $\ket{\psi^{\perp}},\,\ket{\phi^{\perp}}$
are normalized and belong to $\mathcal{H}_{\ket{\chi}}^{\perp}$,
the orthogonal complement of $\ket{\chi}$ in $\mathcal{H}$. Using
(\ref{eq:cration single craus}) we obtain the following condition
\begin{widetext}
\begin{gather}
V\ket{\nu}\otimes\left(\sqrt{c_{1}c_{2}}\ket{\chi}\ket{\chi}+\sqrt{c_{1}}q_{2}\ket{\chi}\ket{\phi^{\perp}}+\sqrt{c_{2}}q_{1}\ket{\psi^{\perp}}\ket{\chi}+q_{1}q_{2}\ket{\psi^{\perp}}\ket{\phi^{\perp}}\right)= \nonumber \\
=\tilde{\mathcal{A}}\left(\ket{\psi^{\perp}},\ket{\phi^{\perp}}\right)\cdot\left(\left(\alpha\sqrt{c_{1}}+\beta\sqrt{c_{2}}\right)\ket{\chi}+q_{1}\alpha\ket{\psi^{\perp}}+q_{2}\beta\ket{\phi^{\perp}}\right)\, \tag{S.25},\label{eq:expansion}
\end{gather}
\end{widetext}
where $\tilde{\mathcal{A}}\left(\ket{\psi^{\perp}},\ket{\phi^{\perp}}\right)=\mathcal{A}\left(\ket{\psi},\ket{\phi}\right)$
for $\ket{\psi},\ket{\phi}$ given by (\ref{eq:vec repr 1}) and (\ref{eq: vect repr 2}).
For the fixed normalized vectors $\ket{\psi_{0}^{\perp}},\ket{\phi_{0}^{\perp}}\in\mathcal{H}_{\ket{\chi}}^{\perp}$
the function 

\begin{equation}
\left(\theta_{1},\theta_{2}\right)\rightarrow\tilde{\mathcal{A}}\left(\mathrm{exp}\left(i\theta_{1}\right)\ket{\psi_{0}^{\perp}},\mathrm{exp}\left(i\theta_{2}\right)\ket{\phi_{0}^{\perp}}\right)\tag{S.26} \label{eq:aux map}
\end{equation}
is a smooth function on a torus $\mathbb{S}_{1}\times\mathbb{S}_{1}$.
It follows from the expression
\begin{equation}
\tilde{\mathcal{A}}\left(\mathrm{exp}\left(i\theta_{1}\right)\ket{\psi_{0}^{\perp}},\mathrm{exp}\left(i\theta_{2}\right)\ket{\phi_{0}^{\perp}}\right)=\frac{f\left(\theta_{1},\theta_{2}\right)}{g\left(\theta_{1},\theta_{2}\right)}\, \tag{S.27},\label{eq:another expression}
\end{equation}
where $f$ and $g$ are smooth and $g\neq0$. Equation (\ref{eq:another expression})
follows from the definition of $\tilde{\mathcal{A}}$ and equations
(\ref{eq:second kraus}), (\ref{eq:simplification2}) and (\ref{eq:phase in A}). 

Now, by inserting $\ket{\psi^{\perp}}=\mathrm{exp}\left(i\theta_{1}\right)\ket{\psi_{0}^{\perp}}$
and $\ket{\phi^{\perp}}=\mathrm{exp}\left(i\theta_{2}\right)\ket{\phi_{0}^{\perp}}$
into (\ref{eq:expansion}), we can view expressions appearing on both
sides of equality (\ref{eq:expansion}) as integrable vector-valued
functions of the pair of angles $\left(\theta_{1},\theta_{2}\right)$.
Using the linearity of $V$ and comparing Fourier coefficients on
both sides of (\ref{eq:expansion}) we obtain that for all $\theta_{1},\theta_{2}\in\left[0,2\pi\right)$

\begin{equation}
\tilde{\mathcal{A}}\left(\mathrm{exp}\left(i\theta_{1}\right)\ket{\psi_{0}^{\perp}},\mathrm{exp}\left(i\theta_{2}\right)\ket{\phi_{0}^{\perp}}\right)=\tilde{\mathcal{A}}\left(\ket{\psi_{0}^{\perp}},\ket{\phi_{0}^{\perp}}\right)\,\tag{S.28}.\label{eq:constance on pairs}
\end{equation}
Moreover, we get

\begin{gather}
V\ket{\nu}\ket{\chi}\ket{\chi}=\tilde{\mathcal{A}}\left(\ket{\psi_{0}^{\perp}},\ket{\phi_{0}^{\perp}}\right)\frac{\alpha\sqrt{c_{1}}+\beta\sqrt{c_{2}}}{\sqrt{c_{1}c_{2}}}\ket{\chi},\label{eq:simplistic bilin} \tag{S.29} \\
V\ket{\nu}\ket{\psi^{\perp}}\ket{\phi^{\perp}}=0\,,\label{eq:simpler-1} \tag{S.30} \\
V\ket{\nu}\ket{\psi_{0}^{\perp}}\ket{\chi}=\tilde{\mathcal{A}}\left(\ket{\psi_{0}^{\perp}},\ket{\phi_{0}^{\perp}}\right)\frac{\alpha}{\sqrt{c_{2}}}\ket{\psi_{0}^{\perp}}\,,\label{eq:more educated} \tag{S.31} \\
V\ket{\nu}\ket{\chi}\ket{\phi_{0}^{\perp}}=\tilde{\mathcal{A}}\left(\ket{\psi_{0}^{\perp}},\ket{\phi_{0}^{\perp}}\right)\frac{\beta}{\sqrt{c_{1}}}\ket{\phi_{0}^{\perp}}\, \tag{S.32}.\label{eq:more educated 2}
\end{gather}
Using (\ref{eq:constance on pairs}) and the fact that (\ref{eq:more educated})
and (\ref{eq:more educated 2}) must hold for all normalized $\ket{\psi_{0}^{\perp}},\ket{\phi_{0}^{\perp}}\in\mathcal{H}_{\ket{\chi}}^{\perp}$
we get that $\tilde{\mathcal{A}}\left(\ket{\psi_{0}^{\perp}},\ket{\phi_{0}^{\perp}}\right)=\mathcal{A}=const.$
We complete the proof by noticing that for constant $\tilde{\mathcal{A}}\left(\ket{\psi_{0}^{\perp}},\ket{\phi_{0}^{\perp}}\right)$
the above conditions uniquely specify the action of a linear map $V$
on every vector $\ket{\Phi}\in\mathbb{C}^{2}\otimes\mathcal{H}\otimes\mathcal{H}$.
\end{proof}

A careful analysis of the above proof 
shows that the protocol also works for all input states satisfying $\mathrm{tr}\left(\mathbb{P}_{\chi}\mathbb{P}_{\psi}\right)=\lambda c_{1}\,,\,\mathrm{tr}\left(\mathbb{P}_{\chi}\mathbb{P}_{\phi}\right)=\lambda c_{2}$,
where $\lambda\in\left(0,\frac{1}{\max\left\{ c_{1},c_{2}\right\} }\right]$.
The appropriate superpositions are then generated with probability
$P'_{succ}=\lambda P_{succ}$. The set of possible inputs for which
a given $\Lambda_{sup}$ works is characterized by the condition $\frac{\mathrm{tr}\left(\mathbb{P}_{\chi}\mathbb{P}_{\psi}\right)}{\mathrm{tr}\left(\mathbb{P}_{\chi}\mathbb{P}_{\phi}\right)}=c$.
Combining this with we uniqueness result we get that it is not possible
to probabilistically generate superpositions (\ref{eq:desired super})
for all input states having nonzero overlap with $\mathbb{P}_{\chi}$.

We now present an explicit protocol that generates the superposition
(\ref{eq:desired super}) with the higher probability of success than
the one given in the proof of Theorem 2 but works
for the fixed coeffitients $\alpha,\beta$ satisfying $\left|\alpha\right|^{2}+\left|\beta\right|^{2}=1$.
Let $\tilde{\Lambda}_{sup}\left(\rho\right)=W\rho W^{\dagger}$, for
a linear mapping $W:\mathcal{H}\otimes\mathcal{H}\rightarrow\mathcal{H}$
defined by $W=W_{2}W_{1}$, where 
\begin{gather}
W_{1}=\frac{\alpha}{\sqrt{c_{1}}}\mathbb{I}\otimes\mathbb{I}+\frac{\beta}{\sqrt{c_{2}}}\mathbb{S}\, \tag{S.33},\\
W_{2}=\mathbb{I}\otimes\bra{\chi}\, \tag{S.34}.
\end{gather}
In the above $\mathbb{S}$ denotes the unitary operator that swaps
between two copies of $\mathcal{H}$ , and $\ket{\chi}$ is a vector
representative of $\mathbb{P}_{\chi}$. The action of $V_{2}$ on
simple tensors is given by 

\begin{equation}
\mathbb{I}\otimes\bra{\chi}\left(\ket x\ket y\right)=\ket x\bk{\chi}x\,, \tag{S.35}
\end{equation}
for all $\ket x,\ket y\in\mathcal{H}$. Explicit computation shows
that for vectors $\ket{\psi},\,\ket{\phi}$ which are vector represent
ants of the input states $\mathbb{P}_{\psi},\,\mathbb{P}_{\phi}$
we have 

\begin{equation}
W\ket{\psi}\ket{\phi}=\alpha\frac{\bk{\chi}{\phi}}{\left|\bk{\chi}{\phi}\right|}\ket{\psi}+\beta\frac{\bk{\chi}{\psi}}{\left|\bk{\chi}{\psi}\right|}\ket{\phi}\,. \tag{S.36}
\end{equation}
which shows that $W\mathbb{P}_{\psi}\otimes\mathbb{P}_{\phi}W^{\dagger}\propto\kb{\Psi}{\Psi}$.
The map $\tilde{\Lambda}_{sup}$ is not normalized i.e. it might happen
that it increases the trace. Since $\tilde{\Lambda}_{sup}$ can be
expressed via a single Krauss operator, it is trace non-increasing
if and only if \cite{Nielsen2010} operator $W$ satisfies $W^{\dagger}W\leq\mathbb{I}\otimes\mathbb{I}$.
We have

\begin{align}
W^{\dagger}W & =\frac{\left|\alpha\right|^{2}}{c_{1}}\mathbb{I}\otimes\kb{\chi}{\chi}+\frac{\left|\beta\right|^{2}}{c_{2}}\kb{\chi}{\chi}\otimes\mathbb{I} + \nonumber \\ & +\frac{\alpha\bar{\beta}}{\sqrt{c_{1}c_{2}}}\mathbb{I}\otimes\kb{\chi}{\chi}\mathbb{S}+\mathbb{S}\frac{\bar{\alpha}\beta}{\sqrt{c_{1}c_{2}}}\mathbb{I}\otimes\kb{\chi}{\chi}\,. \tag{S.37}\label{eq:explicit} 
\end{align}
Explicit computation shows that the maximal eigenvalue of $W^{\dagger}W$
is given by 

\begin{equation}
\lambda_{\max}\left(W^{\dagger}W\right)=\max\left\{ \left|\frac{\alpha}{\sqrt{c_{1}}}+\frac{\beta}{\sqrt{c_{2}}}\right|^{2},\frac{\left|\alpha\right|^{2}}{c_{1}}+\frac{\left|\beta\right|^{2}}{c_{2}}\right\} \,.\label{eq:miximal eigenvalue} \tag{S.38}
\end{equation}
The largest possible $s\in\mathbb{R}_{+}$ such that $s\cdot\Lambda_{sup}$
is trace non-increasing is $s_{max}=\left[\lambda_{\max}\left(W^{\dagger}W\right)\right]^{-1}$.
The probability that $s_{max}\Lambda_{sup}$ will produce the superposition
$\mathbb{P}_{\Psi}$ is given by 
\begin{equation}
\tilde{P}_{succ}=s_{max}\cdot\mathrm{tr}\left[\tilde{\Lambda}_{sup}\left(\mathbb{P}_{\psi}\otimes\mathbb{P}_{\phi}\right)\right]=\frac{\mathcal{N}_{\Psi}^{2}}{\lambda_{\max}\left(W^{\dagger}W\right)}\,,\label{eq:probability of succ 1} \tag{S.39}
\end{equation}
Comparing (\ref{eq:probability of succ 1}) with (15)
and using (\ref{eq:miximal eigenvalue}) we see that $\tilde{P}_{succ}\geq P_{succ}$
if an only if 

\begin{equation}
\frac{1}{c_{1}}+\frac{1}{c_{2}}\geq\max\left\{ \left|\frac{\alpha}{\sqrt{c_{1}}}+\frac{\beta}{\sqrt{c_{2}}}\right|^{2},\frac{\left|\alpha\right|^{2}}{c_{1}}+\frac{\left|\beta\right|^{2}}{c_{2}}\right\} \,,\label{eq:comparison} \tag{S.40}
\end{equation}
for $c_{1,2}\in\left(0,1\right]$ and $\left|\alpha\right|^{2}+\left|\beta\right|^{2}=1$.
Inequality (\ref{eq:comparison}) can be easily checked via elementary
means. Just like in the case of $\Lambda_{sup}$ it is possible to
show that $\tilde{\Lambda}_{sup}$ is defined uniquelly up to scalling.
Consequently, the probability of success $\tilde{P}_{succ}$ given
in (\ref{eq:probability of succ 1}) is the highest possible one,
provided the coeffitients $\alpha,\beta$ are fixed.

\subsection*{Part C: Constructive protocol for superposition of multiple states}

In this part we generalize the protocol presented in the proof of
Theorem 2 to the case of multiple superpositions.

First, we generalize the formula (6).
Assume that we are given a known pure stet $\mathbb{P}_{\chi}$ and
$d$ unknown pure states $\mathbb{P}_{\psi_{1}},\ldots,\mathbb{P}_{\psi_{d}}$
satisfying 

\begin{equation}
\mathrm{tr}\left(\mathbb{P}_{\chi}\mathbb{P}_{\psi_{i}}\right)\neq0\,,\, i=1,\ldots,d\,,\label{eq:multiple condition} \tag{S.41}
\end{equation}
We now introduce the mapping $\left(\mathbb{P}_{\psi_{1}},\ldots,\mathbb{P}_{\psi_{d}}\right)\rightarrow\kb{\Psi}{\Psi}$
that would associate to any sequence of such pure states their superposition.
Let $\ket{\chi}$ be a vector representative of $\mathbb{P}_{\chi}$
and let $\ket{\psi_{i}}$, , be vector representatives
of states $\mathbb{P}_{\psi_{i}}$, $i=1,\ldots,d$ . For a given sequence of complex
coefficients $\alpha_{1},\ldots,\alpha_{d}$ (satisfying $\sum_{i=1}^{d}\left|\alpha_{i}\right|^{2}=1$)
we set

\begin{equation}
\ket{\Psi_{d}}=\sum_{i=1}^{d}\alpha_{i}\frac{\prod_{k\neq i}\bk{\chi}{\psi_{k}}}{\prod_{k\neq i}\left|\bk{\chi}{\psi_{k}}\right|}\ket{\psi_{i}}\,,\label{eq:multiple superposition} \tag{S.42} \ .
\end{equation}
The above formula is a direct
generalization of (6) and analogous
arguments show that $\kb{\Psi}{\Psi}$ is a well-defined function
of states $\mathbb{P}_{\psi_{1}},\ldots,\mathbb{P}_{\psi_{d}}$. This
can be also seen from the explicit formula 

\begin{align}
\kb{\Psi_d}{\Psi_d} & =\sum_{i=1}^{d}\left|\alpha_{i}\right|^{2}\mathbb{P}_{\psi_{i}}+\nonumber  \\
 & +\sum_{i,j:i\neq j}\frac{\alpha_i \bar{\alpha}_j\mathbb{P}_{\psi_{i}}\mathbb{M}_{ij}\mathbb{P}_{\psi_{j}}}{\sqrt{\prod_{k\neq i}\mathrm{tr}\left(\mathbb{P}_{\chi}\mathbb{P}_{\psi_{k}}\right)}\sqrt{\prod_{l\neq j}\mathrm{tr}\left(\mathbb{P}_{\chi}\mathbb{P}_{\psi_{l}}\right)}}\,,\label{eq:full formula multiply} \tag{S.43}
\end{align}
where 

\begin{equation}
\mathbb{M}_{ij}=\mathbb{P}_{\chi}\prod_{k:k\neq i,k\neq j}\left(\mathbb{P}_{\chi}\mathbb{P}_{\psi_{k}}\mathbb{P}_{\chi}\right)\, \ ,\label{eq:auxilary express} \tag{S.44}
\end{equation}
where it is understood that the product of operators indexed by the empty set  is the identity operator.
\begin{remark}
In the context of geometric approaches to quantum mechanics \cite{Chaturvedi2013,Ercolessi2010,Manko2002}
there appears the following superposition rule,

\begin{equation} 
\ket{\Psi_d'}=\sum_{i=1}^{d}\alpha_{i}\frac{\bk{\psi_{i}}{\chi}}{\left|\bk{\psi_{i}}{\chi}\right|}\ket{\psi_{i}}\,.\label{eq:alternativesuperp} \tag{S.45}
\end{equation}
For $d>2$ the mappings 
\[
\left(\mathbb{P}_{\psi_{1}},\ldots,\mathbb{P}_{\psi_{d}}\right)\rightarrow\kb{\Psi_d}{\Psi_d}\,,\,\left(\mathbb{P}_{\psi_{1}},\ldots,\mathbb{P}_{\psi_{d}}\right)\rightarrow\kb{\Psi_d'}{\Psi'_d} \ , \tag{S.46}
\]
differ from each other. It would be interesting to explore the geometric
significance of the superposition  rule used by us. Also, It would be interesting to explore if states of the form \eqref{eq:alternativesuperp} can be generated by processes allowed by quantum mechanics \end{remark}
\begin{thm}
\label{protoco qdit} Let $\mathbb{P}_{\chi}$ be a fixed pure state
on Hilbert space $\mathcal{H}$ and let $d>1$ be a natural number.
There exist a CP map $\Lambda_{sup}^{d}\in\mathcal{CP}\left(\mathbb{C}^{d}\otimes\mathcal{H}^{\otimes d},\mathcal{H}\right)$
such that for all pure states $\mathbb{P}_{\psi_{i}}\,,\, i=1,\ldots,d$,
on $\mathcal{H}$ satisfying 
\begin{equation}
\mathrm{tr}\left(\mathbb{P}_{\chi}\mathbb{P}_{\psi_{i}}\right)=c_{i}\,,\, i=1,\ldots,d\,,\label{eq:conditions multi} \tag{S.47}
\end{equation}
 we have
\begin{equation}
\Lambda_{sup}^{d}\left(\mathbb{P}_{\nu}\otimes\bigotimes_{i=1}^{d}\mathbb{P}_{\psi_{i}}\right)\propto\kb{\Psi_{d}}{\Psi_{d}}\,,\label{eq:condition superp multi} \tag{S.48}
\end{equation}
where
\begin{equation}
\mathbb{P}_{\nu}\,,\,\ket{\nu}=\sum_{i=1}^{d}\alpha_{i}\ket i,\label{eq:qdit unknown} \tag{S.49}
\end{equation}
is an unknown qdit state and the vector $\ket{\Psi_{d}}$ is given
by (\ref{eq:multiple superposition}). Moreover, a CP map
$\Lambda_{sup}^{d}$  realising (\ref{eq:condition superp multi}) is unique up scaling. \end{thm}
\begin{proof}
We present an explicit protocol that realizes (\ref{eq:condition superp multi})
which is analogous to the one given in the proof of Theorem 2.
Let us first define and auxiliary normalized qdit vector 
\begin{equation}
\ket{\mu_{d}}=\mathcal{C}_{d}\sum_{i=1}^{d}\frac{1}{\prod_{k\neq i}\sqrt{c_{i}}}\ket i\,, \tag{S.50}
\end{equation}
where $\mathcal{C}_{d}$ is a normalization constant.
We set $\Lambda_{sup}^{d}=\Lambda_{4}\circ\Lambda_{3}\circ\Lambda_{2}\circ\Lambda_{1}$,
where 
\begin{gather}
\Lambda_{1}\left(\rho\right) = V_{1}\rho V_{1}^{\dagger}\,,\, V_{1}=\sum_{i=1}^{d}\kb ii\otimes\mathbb{S}_{1,i}\, , \tag{S.51}\\
\Lambda_{2}\left(\rho\right) = V_{2}\rho V_{2}^{\dagger}\,,\, V_{2}=\mathbb{I}_{d}\otimes\mathbb{I}\otimes\left(\kb{\chi}{\chi}^{\otimes d-1}\right)\, \tag{S.52},\\
\Lambda_{3}\left(\rho\right) = V_{3}\rho V_{3}\,,\, V_{3}=\mathbb{P}_{\mu_{d}}\otimes\mathbb{I}^{\otimes d}\, \tag{S.53},\\
\Lambda_{4}\left(\rho\right) = \mathrm{tr}_{1,3,\ldots,d+1}\left(\rho\right) \tag{S.54}\,.
\end{gather}
In the above $\mathbb{S}_{1,i}:\mathcal{H}^{\otimes d}\rightarrow\mathcal{H}^{\otimes d}$
denotes the unitary operator that swaps between first and i'th copy
of the Hilbert space$\mathcal{H}^{\otimes d}$, $\mathbb{I}_{d}$
and $\mathbb{I}$ are identity operators on $\mathbb{C}^{d}$ and
$\mathcal{H}$ respectively, $\mathrm{tr}_{1,3,\ldots,d+1}\left(\rho\right)$
is the partial trace over all except for the second factor in the
tensor product $\mathbb{C}^{d}\otimes\mathcal{H}\otimes\mathcal{H}^{\otimes d-1}$.
Operation $\Lambda_{sup}^{d}$ is manifestly completely positive and
trace non-increasing. Direct calculation shows that under the assumed
conditions (\ref{eq:condition superp multi}) indeed holds. 
The proof of uniqueness of the map $\Lambda_{sup}^{d}$ is analogous to the one given in the case $d=2$ (covered by Theorem 2).
\end{proof}

\subsection*{Part D: Creation of superpositions of results of subroutines of quantum computations run in parallel}

Here we show how to implement the protocol coherently  superposing results of subroutines of quantum computation in the standard quantum circuit formalism. Assume we want to superpose states $\mathbb{P}_{\psi}$, $\mathbb{P}_{\phi}$ that correspond to application of some quantum circuts on $N$ qbits (the Hilbert space of the system $\mathcal{H}=\left(\mathbb{C}^2\right)^{\otimes N}$) , 
\begin{equation}
\mathbb{P}_{\psi}=U\kb{x}{x}U^{\dagger}\ , \ \mathbb{P}_{\phi}=V\kb{y}{y}V^{\dagger}\ , \tag{S.55}
\end{equation} 
where $\kb{x}{x}$, $\kb{y}{y}$ are classical states encoding the cinput to the the computation,
 \begin{equation} \label{input}
\ket{x}=\ket{x_1}\otimes \ldots \otimes \ket{x_N} \ , \ket{y}=\ket{y_1}\otimes \ldots \otimes \ket{y_N} \ , \tag{S.55}
\end{equation}
and $U,V$ are unitary operators on $\left(\mathbb{C}^2\right)^{\otimes N}$.  We introduce the auxiliary qubit that will allow us to encode the state $\mathbb{P}_\chi$ without altering the computation (from now on we will consider the Hilbert space $\tilde{\mathcal{H}}=\mathbb{C}^2 \otimes \mathcal{H}$). We set $\ket{\chi}=\ket{0}^{\otimes N+1}$ and we introduce new initial states as projectors onto vectors
\begin{gather}
\ket{\tilde{x}}=\frac{1}{\sqrt{2}}\left(\ket{\chi}+\ket{1}\ket{x}\right) \ , \nonumber \\
\ket{\tilde{y}}=\frac{1}{\sqrt{2}}\left(\ket{\chi}+\ket{1}\ket{y}\right) \ . \label{eq:encoding1} \tag{S.56}
\end{gather}
By applying the controlled versions of the unitaries $U,V$,
\begin{equation} \label{eq:encunit}
\mathcal{C}\left({U}\right)=\kb{0}{0}\otimes \mathbb{I}+\kb{1}{1}\otimes U \ , \mathcal{C}\left({V}\right)=\kb{1}{1}\otimes \mathbb{I}+\kb{1}{1}\otimes V \ ,  \tag{S.57}
\end{equation}
we obtain the states represented by the vectors 
\begin{gather}
\ket{\tilde{\psi}}=\frac{1}{\sqrt{2}}\left(\ket{\chi}+\ket{1}U\ket{x}\right) \ , \nonumber \\
\ket{\tilde{\phi}}=\frac{1}{\sqrt{2}}\left(\ket{\chi}+\ket{1}V\ket{y}\right) \ . \label{eq:encodres} \tag{S.58}
\end{gather}
These vectors are exactly of the form given by (16) and thus we can apply to them protocol from the main text of the manuscript (keep in mind that we set $\ket{\chi}=\ket{0}^{\otimes N+1}$) that with the probability 
\begin{equation}
P_{succ}=\frac{1}{4}\left(1+\frac{1}{4}\left\|U\ket{x}+V\ket{y}\right\|^2\right) \ \tag{S.59} ,
\end{equation}
will produce a state having a vector representative
\begin{equation} \label{eq:compsuperpos}
\ket{\Psi}=\frac{1}{\sqrt{P_{succ}}}(\ket{0}^{\otimes N+1}+\frac{1}{2}\ket{1}\left[U\ket{x}+V\ket{y}] \right) \ . \tag{S.60}
\end{equation}
Note that from the state \eqref{eq:compsuperpos} it is possible to extract (by postselecting with respect to obtaining the result "1" in the auxiliary qbit) the state encoding the superposition $U\ket{x}+V\ket{y}$ in the computational register.

Let us discuss the method presented above. First of all, in order to be able to create the desired superpositions we need to encode input states in the extended space (see \eqref{eq:encoding1}) and use the controlled versions of the gates $U,V$ (see \eqref{eq:encunit}). The controlled versions of the unitary gate are defined up to a phase standing next to the  unitary which is controlled  \cite{Bisio2015} and therefore we can always add an additional phase in front of the second terms in the sums \eqref{eq:encodres}. However, this is not a problem as we can always decode (probabilistically)  the original computation from states of the form \eqref{eq:encodres}. Another possible problem may come form the necessity of implementing the controlled versions of gates $U,V$. This can be always done if one knows the classical description of these gates. In particular, assume that the $N$ qbit gate $U$ can be decomposed as a sequence of basic gates $U_1,U_2,\ldots,U_k$ ,
\begin{equation}
U=U_k \circ \ldots U_2 \circ U_1 \tag{S.61} \ .
\end{equation}
Then, a controlled version of of this gate can be obtained by composing controlled versions of the basic gates,
\begin{equation} \label{eq:cotrcomp}
\mathcal{C}\left({U}\right)=\mathcal{C}\left({U_k}\right) \circ \ldots \mathcal{C}\left({U_2}\right) \circ \mathcal{C}\left({U_1}\right) \tag{S.62} \ .
\end{equation}
For the graphical illustration of the preperation of the state $\mathbb{P}_{\tilde{\psi}}$  see Figure \ref{fig:circuit}.

\begin{figure}[h]
%\scalebox{0.40}{\includegraphics{few-devices-fig.pdf}}
\scalebox{0.5}{\includegraphics[trim= 0cm 0cm 6cm 0cm, clip=true]{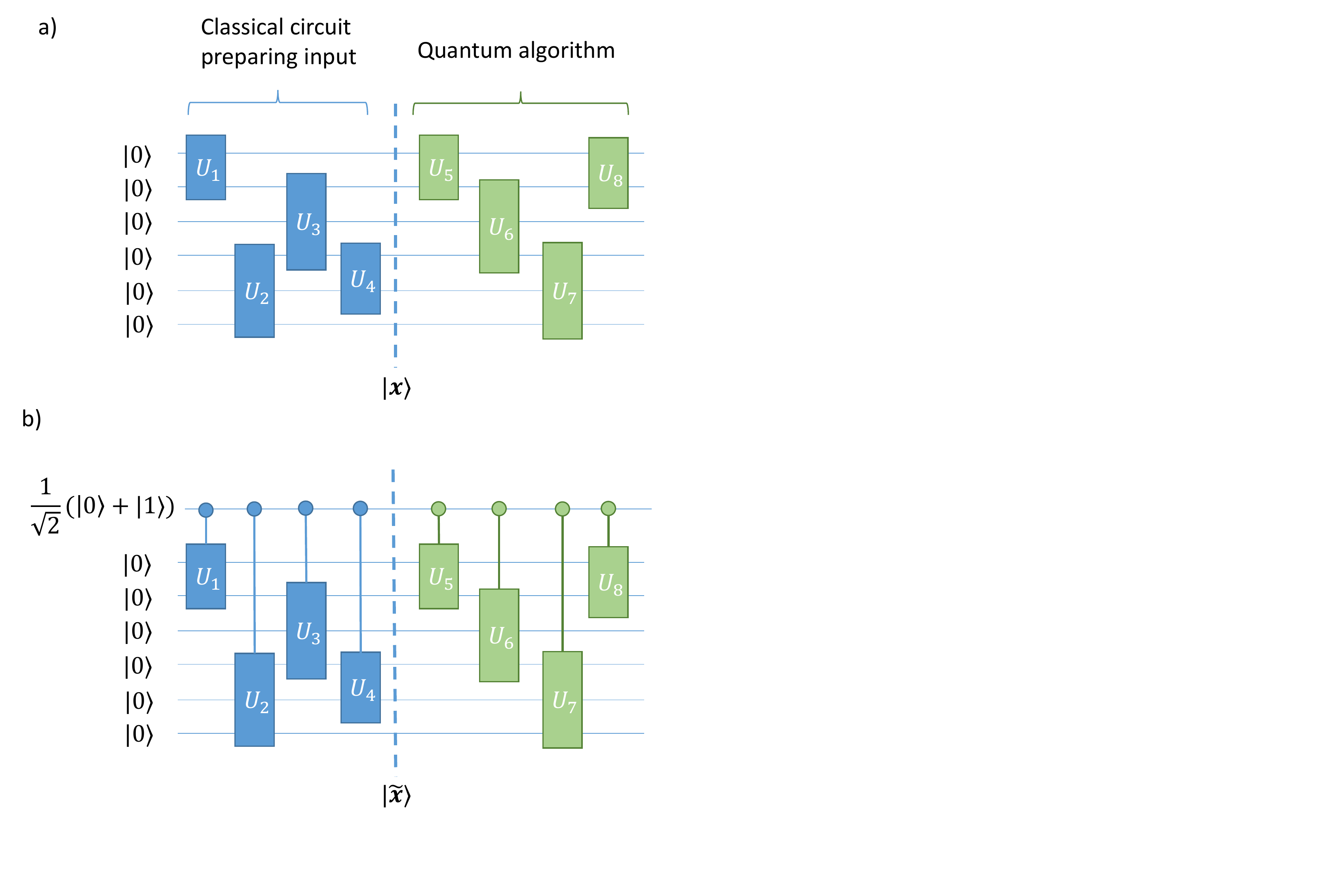}}
\centering
  \caption{Computing in the subspace perpendicular to $\ket{chi}$. (a) original circuit, including preparation of classical input $\ket{x}$ 
and the quantum algorithm producing $U\ket{x}$. (b) the new circuit uses previous gates controlled by ancillary qubit. 
The blue part of the circuit prepares $\ket{\tilde{x}}=\frac{1}{\sqrt2}\left(\ket{0}\ket{0}^{\otimes N}+|1\>\ket{x} \right)$. 
The green part runs computing on $\ket{x}$ resulting in 
$\ket{\tilde{\psi}}=\frac{1}{\sqrt2}\left(\ket{0}\ket{0}^{\otimes N}+\ket{1}U\ket{x} \right)$
The reference vector is $\chi=\ket{0}^{\otimes N+1}$. }
\label{fig:circuit}
\end{figure}

To sum up, the protocol presented above creates superpositions of results of subroutines of quantum computations run in parallel with probability $P_{succ}\geq\frac{1}{4}$. In order to implement the method one has to know what quantum subroutines are implemented. However, the inputs of the computations can be arbitrary (there are no constrains on the classical input states $\kb{x}{x}$ and  $\kb{y}{y}$).

\end{document}